\documentclass[11pt,english]{article}

\pdfoutput=1

\usepackage{hyperref}
\usepackage{verbatim}
\usepackage{amsfonts}
\usepackage{amsmath,amssymb,amsthm}
\usepackage{multirow}
\usepackage{url}
\usepackage{bbm}
\usepackage{dsfont}

\usepackage{courier}
\usepackage{fullpage}

\usepackage{epsfig}
\usepackage[usenames]{color}
\usepackage{xspace}
\usepackage{algorithm}
\usepackage[noend]{algpseudocode}
\usepackage{times}
\usepackage{enumerate} 

\usepackage[title]{appendix}

\newtheorem{theorem}{Theorem}[section]
\newtheorem{lemma}[theorem]{Lemma}
\newtheorem{proposition}[theorem]{Proposition}
\newtheorem{fact}[theorem]{Fact}
\newtheorem{corollary}[theorem]{Corollary}

\newtheorem{claim}[theorem]{Claim}

\newtheorem{definition}[theorem]{Definition}

\newtheorem{construction}[theorem]{Construction}

\newtheorem*{lemma-main}{Lemma~\ref{lemma:main}}
\newtheorem*{prop-restriction}{Proposition~\ref{prop:restriction}}

\renewcommand{\paragraph}[1]{\noindent {\bf #1}}
\newcommand{\Exp}{\operatornamewithlimits{\mathbb{E}}}

\newcommand{\F}{\mathbb{F}}
\newcommand{\C}{\mathbb{C}}
\newcommand{\R}{\mathbb{R}}
\newcommand{\Z}{{\mathbb Z}}
\newcommand{\N}{{\mathbb N}}

\newcommand{\polylog}[1]{\mathrm{polylog}(#1)}
\newcommand{\poly}[1]{\mathrm{poly}(#1)}

\newcommand {\etal} {\textit{et al.}\xspace}

\newcommand{\gran}[1]{\operatorname{gran}(#1)}

\newcommand{\supp}[1]{\operatorname{supp}(#1)}

\newcommand{\vspan}[1]{\operatorname{span}(#1)}

\newcommand\dcc{\mbox{$\sf {D^{CC}}$}\xspace}

\newcommand{\ignore}[1]{}

\begin{document}
\title{A Generalization of a Theorem of Rothschild and van Lint}

\author{Ning Xie\thanks{Florida International University, Miami, FL 33199, USA. Email: {\texttt nxie@cis.fiu.edu}} 
\and Shuai Xu\thanks{Case Western Reserve University, Cleveland, OH 44106, USA. Email: {\texttt sxx214@case.edu}.} 
\and Yekun Xu\thanks{Florida International University, Miami, FL 33199, USA. Email: {\texttt yxu040@fiu.edu}.}}

\date{}

\setcounter{page}{0}
\maketitle

\begin{abstract}
A classical result of Rothschild and van Lint asserts that if
every non-zero Fourier coefficient of a Boolean function $f$ over $\mathbb{F}_2^{n}$ has the same absolute value, namely $|\hat{f}(\alpha)|=1/2^k$ for every $\alpha$ in the Fourier support of $f$, then $f$ must be the indicator function of some affine subspace of dimension $n-k$.
In this paper we slightly generalize their result. Our main result shows that, roughly speaking, Boolean functions whose Fourier coefficients take values in the set $\{-2/2^k, -1/2^k, 0, 1/2^k, 2/2^k\}$ are indicator functions of two disjoint affine subspaces of dimension $n-k$ or four disjoint affine subspace of dimension $n-k-1$. 
Our main technical tools are results from additive combinatorics which offer tight bounds on the affine span size of a subset of $\mathbb{F}_2^{n}$ when the doubling constant of the subset is small.
\end{abstract}

\newpage

\section{Introduction}\label{Sec:intro}
One of the most fruitful approaches in functional analysis is to represent functions as sums of simple and well-structured objects, 
such as sine wave functions and polynomials. Such representations often provide additional insights on the combinatorial structures of or 
complexity measures associated with the subjects under consideration.
This paradigm in theoretical computer science has witnessed
harmonic analysis on the cube, or the discrete Fourier transform of Boolean functions, emerged in the past three decades 
as a powerful and versatile tool that finds numerous applications in complexity theory (such as PCP and circuit complexity), 
property testing, learning, cryptography, coding theory, social choice theory and others; see~\cite{Odo14} for a comprehensive survey. 

Fourier coefficients and function values are two equivalent ways to represent a function.
That is, the Fourier spectrum of a function completely determines the function-value at any point on the cube. 
However, knowing only the \emph{values} of the Fourier spectrum but without the information of the locations
of these values in the Fourier space in general leaves the function undetermined to a large extent, even restricted to Boolean functions. 
To see this, consider the following examples.
Generally speaking, we view two Boolean functions as the same function if they are \emph{isomorphic}. 
More formally, we say that two Boolean functions $f,g: \F_2^n \to \{0,1\}$ are \emph{isomorphic} to each other if there is an invertible linear
transformation $L: \F_2^n \to \F_2^n$ such that $g(x)=Lf(x)$ for every $x\in \F_2^n$, where $Lf(x):=f(Lx)$.
Now consider the following two families of Boolean functions $\{f_k: \F_2^k \to \{0,1\}\mid k\in \N, k \geq 3\}$
and $\{g_k: \F_2^k \to \{0,1\}\mid k\in \N, k \geq 3\}$, with the Fourier expansions of
$f_k(x)=\frac{3}{4}-\frac{1}{4}\chi_{\{1\}}(x)-\frac{1}{4}\chi_{\{2\}}(x)-\frac{1}{4}\chi_{\{1,2\}}(x)$
and $g_k(x)=\frac{3}{4}-\frac{1}{4}\chi_{\{1,2\}}(x)-\frac{1}{4}\chi_{\{1,3\}}(x)-\frac{1}{4}\chi_{\{2,3\}}(x)$.
One can check easily that both $f_k$ and $g_k$ are indeed Boolean functions and the multisets of non-zero Fourier coefficients are both
$\{\frac{3}{4}, -\frac{1}{4}, -\frac{1}{4}, -\frac{1}{4}\}$. On the other hand,
the Fourier dimension --- dimension of the subspace spanned by vectors at which the function's Fourier coefficients are non-zero ---
of $f_k$ is $2$ while the Fourier dimension of $g_k$ is $3$. Since the Fourier spectrum transforms according to $(L^T)^{-1}$
when the function undergoes the linear transformation $L$, 
it follows that there is no invertible linear transformation $L$ that maps $f_k$ to $g_k$, i.e.
they are not isomorphic to each other. Another such example is the class of \emph{address functions} $f_n: \F_2^n \to \{-1,1\}$, where $n=k+2^k$ for some positive integer $k$,
together with the class of functions $g_n: \F_2^n \to \{-1,1\}$ formed by tensoring some \emph{bent function} on $2k$-bits with a $\delta$-function on $n-2k$ bits. Then both $f_n$ and $g_n$
have $2^{2k}$ non-zero Fourier coefficients, with $2^{2k-1}+2^{k-1}$ of them taking value $1/2^k$ and $2^{2k-1}-2^{k-1}$ of them taking value $-1/2^k$; moreover, since the Fourier dimension of $f_n$ is $n$ and the Fourier dimension of $g_n$ is $2k<n$, these two functions are not isomorphic to each other.

Nevertheless, there are a few exceptions to the general phenomenon in the sense that knowing only the values
of the Fourier spectrum completely determine the Boolean function, up to an isomorphism. One such example is the indicator function of an affine subspace, which enjoys a very simple Fourier spectrum. Specifically,
if $f$ is the indicator function of an affine subspace in $\F_2^n$ of dimension $n-k$,
then it is straightforward to check that every non-zero
Fourier coefficient of $f$ is either $1/2^k$ or $-1/2^k$. What about the converse? 
Namely, if we know that the non-zero Fourier coefficients of a Boolean function all have magnitude $1/2^k$, 
then what can be said about the function?

\subsection{Rothschild and van Lint Theorem}
Rothschild and van Lint~\cite{RV74} (see also Chapter 13, Lemma 6 in~\cite{MS77}) proved the following theorem:
\begin{theorem}\label{thm:RV}
Let $n\geq 1$ and $0 \leq k \leq n$.
Let $f=\mathds{1}_S$ be the indicator function of a set $S\subseteq \F_2^n$ of size $|S|=2^{n-k}$.
If for every $\alpha\in  \F_2^n$, $|\hat{f}(\alpha)|$ is equal to either zero or $1/2^k$,
then $S$ is an affine subspace of dimension $n-k$.
\end{theorem}

In other words, Rothschild and van Lint Theorem shows that, up to an invertible linear transform, we have a complete characterization when
the Fourier coefficients of a \emph{Boolean} function are all from the set $\{-1/2^k, 0, 1/2^k\}$:
the Boolean function must be the indicator of some affine subspace of co-dimension $k$.

A natural question is: how far can we extend such a nice characterization in terms of the values of Fourier coefficients only?
Following~\cite{GOS+11}, for a rational number $x$, 
the \emph{granularity} $\gran{x}$ of $x$ is defined to be the least nonnegative integer $k$ such that
$x=m/2^k$, where $m$ is an (odd) integer.
A function $\F_2^n \to \R$ is said to be \emph{$k$-granular} if the maximum granularity of its Fourier coefficients is $k$ --- that is,
$k=\max_{\alpha} \{\gran{\hat{f}(\alpha)}$\}.
For a Boolean function, its granularity is known to be intimately correlated with its \emph{Fourier sparsity}~\cite{GOS+11} --- 
the number of non-zero Fourier coefficients; see discussion in Section~\ref{sec:intro_motivations} for more details.
Therefore, one can view Rothschild and van Lint Theorem as a characterization of 
$k$-granular Boolean functions with minimum support size (that is, $\hat{f}(\mathbf{0})=|\{x: f(x)=1\}|/2^n=1/2^k$).

\subsection{Our results}\label{sec:intro_results}
In this work, we slightly generalize Rothschild and van Lint Theorem to give a complete characterization of
$k$-granular Boolean functions of support size $2^n\cdot 2/2^k=2^{n-k+1}$. 
Roughly speaking, our main theorem is the following:

\begin{theorem}[Informal statement]
For large enough integers $n \geq k$,
if a Boolean function $f:\F_2^n \to \{0,1\}$ has all its Fourier coefficients
in the set $\{0, \frac{\pm 1}{2^k}, \frac{\pm 2}{2^k}\}$,
then $f$ is the indicator function of disjoint union of two affine subspaces of dimension $n-k$.
\end{theorem}

Our Main Theorem is based on the following Main Lemma, which deals with the general case of $k\geq 5$,
together with case analysis\footnote{The need for a nasty case analysis 
stems from a key lemma in the proof, namely Lemma~\ref{lemma:2B}, which holds
only when $k\geq 5$.} for small values of $k$.
\begin{lemma}[Main]\label{lemma:main}
Let $k\geq 5$ and $n\geq k$ be integers.
Let $f: \F_2^n \to \{0,1\}$ be a Boolean function such that $\hat{f}(\mathbf{0})=1/2^{k-1}$ and any other Fourier coefficients are either zero
or equal to $\pm \frac{1}{2^k}$, then $f$ is the indicator function of a disjoint union of two dimension $n-k$
affine subspaces.
\end{lemma}

\subsection{Proof overview and our techniques}\label{sec:intro_techniques}
The original form of Rothschild and van Lint Theorem was stated to characterize subspaces in affine geometry and projective geometry.
For completeness and more importantly, because the first step in our proof of the main theorem follows a similar strategy,
we present a slightly different proof using the notation of Fourier analysis.

\medskip 
\paragraph{A proof of Rothschild and van Lint Theorem.}
We prove the theorem by induction on $n$. It is trivial to see that the theorem holds for $n=1$ (for both $k=0$ and $k=1$).
Let $n\geq 2$. Clearly there is nothing to prove for $k=0$ and $k=n$, so we assume $0<k<n$.
Note that $\hat{f}(\mathbf{0})= |S|/2^n = 1/2^k$, then by Parseval's identity, there exists a non-zero $\alpha$ such that
$\hat{f}(\alpha)=1/2^k$ or $-1/2^k$. Assume that $\hat{f}(\alpha)=1/2^k$ and the case of $\hat{f}(\alpha)=-1/2^k$ is similar.
Applying an invertible linear transform $L$ that maps $\alpha$ to $e_1$, where $e_1$ stands for the standard basis vector $(1,0,\ldots, 0)$.
Note that both the Fourier spectrum of $f$ and any affine subspace are invariant under invertible linear transformations,
hence it suffices to argue about $g:=Lf$.
Now we have $\hat{g}(\mathbf{0})=\hat{g}(e_1)=1/2^k$. 
Applying a linear restriction over the first bit of the input to get sub-functions
$g_0$ and $g_1$ (see Proposition~\ref{prop:restriction} in Appendix~\ref{sec:restriction} for details). 
By \eqref{eqn:subfunction},
$\hat{g}_1(\mathbf{0})=\hat{g}(\mathbf{0})-\hat{g}(e_1)=0$, which implies that $g_1$ is the zero-function.
This implies that $S$ is completely contained in the support of $g_0$ and moreover, by \eqref{eqn:original_function},
$\hat{g}_0(\beta)=2\hat{f}(0, \beta)$ for every $\beta \in \F_2^{n-1}$.
In other words, $g_0$ is a Boolean function over $\F_2^{n-1}$ and $|\hat{g}(\beta)|$ is equal to either zero or $1/2^{k-1}$,
therefore the induction hypothesis applies to $g_0$. 
It follows that $S$ is an affine subspace of dimension
$n-1-(k-1)=n-k$. This completes the proof of Theorem~\ref{thm:RV}.

\medskip 
\paragraph{Reducing the dimension of the function domain.}
The proof of the Main Theorem is much more involved than that of Rothschild and van Lint Theorem.
In fact, the proof we described above of Theorem~\ref{thm:RV} is the first step toward proving the main theorem.
The reduction step in the proof of Theorem~\ref{thm:RV} can be regarded as reducing the dimension of function domain while
keeping all the support of the function. Equivalently, one may view the reduction step as decomposing the original function
$f$ as a \emph{tensor product} between a ``core-function'' $g$ and a ``$\delta$-function'' $h$ (see Section~\ref{sec:prelim}
for definition of tensor product of Boolean functions). Namely, $f(x,y)=g(x)\otimes h(y)$,
where $h:\F_2^m \to \{0,1\}$ is the $\delta$-function: $h(y)=1$ if $y=0^m$ and $h(y)=0$ for all other vectors.
That is, $f$ is ``reduced'' to a core-function $g$ with dimension $n-m$. To this end, we say
a function $f:\F_2^n \to \{0,1\}$ is \emph{reducible} if there exists an invertible linear transformation $L$ such that
$Lf$ can be decomposed as the tensor product of a function
$g:\F_2^{n-m} \to \{0,1\}$ and a $\delta$-function $h$ over $\F_2^m$ with $m\geq 1$.
$f$ is said to be \emph{irreducible} if $f$ is not reducible.\footnote{
To put it differently, a function $f$ defined on $\F_2^n$ is irreducible if and only if the minimum dimension of the affine subspace containing the support of $f$ is $n$.}
Now we are ready to present our Main theorem more precisely.

\begin{theorem}[Main]\label{thm:main}
Let $k\geq 1$, $n>k$ be two integers, and let $f:\F_2^n \to \{0,1\}$ be a non-trivial\footnote{A Boolean function is \emph{trivial} if $f\equiv 0$ or $f\equiv 1$.}
Boolean function with all its Fourier coefficients taking values in $\{0, \frac{\pm 1}{2^k}, \frac{\pm 2}{2^k}\}$. 
Then we have the following complete characterization 
\begin{itemize}
\item If $\hat f(\mathbf{0})=\frac{1}{2^k}$, 
then $f$ is the indicator function of an affine subspace of dimension $n-k$ (Rothschild and van Lint Theorem);
\item If $\hat f(\mathbf{0})=\frac{1}{2^{k-1}}$ and $f$ is irreducible, then $f$ is either the indicator function of disjoint union of two affine subspaces of dimension $n-k$, 
or the indicator function of disjoint union of four affine subspaces of dimension $n-k-1$. Moreover, the latter case is only possible when $k=4$.
\end{itemize}
\end{theorem}

Back to our problem, since $\hat{f}(\mathbf{0})=1/2^{k-1}$, it is easy to see that whenever there is a non-zero $\alpha$
such that $|\hat{f}(\alpha)|=1/2^{k-1}$, we can restrict $f$ either to the subspace $\langle \alpha, x\rangle = 0$
or to the affine subspace $\langle \alpha, x\rangle = 1$ while keeping the entire support of $f$.
We repeat this process until we reach a Boolean function $f$ with $\hat{f}(\mathbf{0})=1/2^{k-1}$ and
all other non-zero Fourier coefficients have magnitude $1/2^k$.

\medskip 
\paragraph{Additive structures of the Fourier spectrum.}
The starting point of our main argument is the following well-known \emph{characterization} of Boolean functions in terms of their
Fourier spectra: a function $f:\F_2^n \to \R$ on the cube is Boolean if and only if
\[
\hat{f}(\alpha)=\sum_{\beta\in \F_2^n}\hat{f}(\beta)\hat{f}(\alpha+\beta)
\] 
holds for every $\alpha \in \F_2^n$.
Our main observation is that, since the non-zero Fourier coefficients $f$ can take only two values when $f$ is irreducible, 
denoting $A:=\{\alpha \mid \hat{f}(\alpha)=1/2^k\}$ and $B:=\{\beta \mid \hat{f}(\beta)=-1/2^k\}$,
then these two sets --- viewed as subsets of abelian group $\F_2^n$ ---  must exhibit strong \emph{additive structures}.
Indeed, one can show that $B+B \subseteq A\cup \{\mathbf{0}\}$ and consequently $|B+B|/|B|\leq (1+|A|)/|B|$.

What can be said about a set $B$ if its \emph{doubling constant} $K:=|B+B|/B$ is small? This is a classical problem extensively studied 
in additive combinatorics. Additive combinatorics is a burgeoning mathematics sub-area which finds
exciting applications in theoretical computer science in recent years~\cite{Sam07,BR15,BLR14,BDL13,ADL18}.
Green and Tao~\cite{GT09b} proved that, when the underlying ambient group is $\F_2^n$, then $B$ is 
contained in a subspace of size $2^{2K+O(K\log{K})}|B|$, which is asymptotically optimal.
Unfortunately, such \emph{asymptotic} ``high end'' bounds are not accurate enough to be useful for our problem. 
In fact, we make crucial use of a ``low end'' additive combinatorics result of Even-Zohar~\cite{Eve12},
which provides tight bounds on the size of affine span of $B$ in terms of its doubling constant.
It is worth noting that all aforementioned applications of additive combinatorics in theoretical computer science
employ theorems regarding \emph{asymptotic} behaviors of certain combinatorial objects. 
We hope researchers may find further applications of such ``low end'' additive combinatorics results in other places.

\subsection{Motivations and related work}\label{sec:intro_motivations}
To the best of our knowledge, besides the work of Rothschild and van Lint, 
there is no previous structural result on Boolean functions 
in terms the \emph{magnitudes} of their Fourier coefficients only.
Friedgut~\cite{Fri98} showed that if the total influence of a Boolean function is small, 
then it is close to some junta --- a function that depends only on a bounded number of variables.
Friedgut {\etal}~\cite{FKN02} studied Boolean functions whose Fourier mass are concentrated on the lowest
two levels and proved that such functions are close to parity functions or negations of parity functions.
For a special class of Boolean functions, the so-called \emph{linear threshold functions},
a celebrated result of Chow~\cite{Cho61} states that these functions are completely
determined by their lowest two level Fourier coefficients;  
see~\cite{OS11,DDFS14} for recent robust versions as well as algorithmic versions of Chow's theorem.
Note that all previous structural theorems mentioned above, except Chow's, 
are ``robust'' in the following sense: 
the structural results are robust against small perturbations in the Boolean function's 
Fourier spectrum. Our main result is automatically robust: by Parseval's identity,
small distance in Fourier spectrum implies small distance in function space; 
consequently, any Boolean function whose Fourier coefficients are close to being in the form stated 
in our Main Theorem must also be close to having the affine subspace structures asserted in the theorem.

Apart from studying to what extent can the values of Fourier coefficients themselves determine a Boolean function,
an important motivation of this research is to study the behaviors of \emph{Fourier sparse} Boolean functions~\cite{GOS+11}.
Gopalan \etal~\cite{GOS+11} proved that, if a Boolean function $f$ has only $s$ non-zero Fourier coefficients, then
every Fourier coefficient of $f$ is of the form $m/2^k$, where $m$ is an integer and $k/2 \leq \log{s} \leq k$.
That is, the granularity and Fourier sparsity of a Boolean function are, up to a constant factor, identical.
Our result may be regarded as characterizing Boolean functions of Fourier granularity $k$ 
when all Fourier coefficients of $f$ are between $-2/2^k$ and $2/2^k$.

Probably the most prominent open problem in communication complexity is the so-called \emph{Log-rank Conjecture}
proposed by Lov{\'a}sz and Saks~\cite{LS88}, which asserts that the deterministic communication complexity 
of any $F: \F_2^n \times \F_2^n \to \{0,1\}$, $\dcc(F)$, 
is upper bounded by a polynomial of the logarithm of the rank of the communication matrix $M_F = [F(x,y)]_{x,y}$,
where the rank is taken over the reals.
Even after more than 30 years of extensive study, we are still very far from resolving it;
the current best bound is Lovett's $\dcc(F)=O(\sqrt{r}\log{r})$~\cite{Lov14}, where $r$ is the rank of $M_F$.
Recently, studying the Log-rank conjecture for a special class of two-party functions, the so-called \emph{XOR functions},
has attracted much attention~\cite{ZS10,TWXZ13,STV17,HHL18,TXZ16,LZ17,CP18}.
The corresponding conjecture for this special class of functions is sometimes called \emph{Log-rank XOR conjecture}.
Specifically, $F$ is an XOR function if there exists an $f: \F_2^n \to \{0,1\}$ such that
for all $x$ and $y$, $F(x,y) = f(x+y)$.
The beautiful connection between the Log-rank XOR conjecture and Fourier analysis of Boolean functions is that, 
if $F$ is an XOR function, then the rank of $M_{F}$ is just the Fourier sparsity of $f$~\cite{BC99}.
Moreover, it is now known that resolving the Log-rank XOR conjecture is equivalent to finding a \emph{parity decision tree}
of depth $\polylog{s}$, or $\poly{k}$ for any Boolean function $f$~\cite{ZS10,TWXZ13,HHL18}, where $s$ is the Fourier sparsity and 
$k$ is the granularity of $f$.

The parity kill number of a Boolean function $f$ is defined as
\[
C_{\oplus, \min}(f) := \min \{\text{co-dim}(S) \mid \text{$S$ is an affine subspace on which $f$ is constant}\}
\]
Tsang \etal~\cite{TWXZ13} demonstrated that, to resolve the Log-rank XOR conjecture, it is sufficient to prove that
the kill number of any Boolean function $f$ is upper bounded by $\polylog{s}$ or $\poly{k}$. 
See~\cite{OST+14,CMS19} for recent developments on constructing Boolean functions with large kill numbers. 
Our main result can be regarded as showing that any Boolean function with granularity $k$ and $\hat{f}(\mathbf{0})\leq 2/2^k$
has kill number at most $k+1$. In fact, by induction on $m$ and folding $\hat{f}(\mathbf{0})$ with any other non-zero Fourier coefficient, we immediately have the following corollary.
\begin{corollary}\label{cor:kill}
Let $f: \F_2^n \to \{0,1\}$ be a Boolean function with granularity $k$ and $\hat{f}(\mathbf{0})=m/2^k$.
Then the kill number of $f$ is at most $k+m-1$.
\end{corollary}
Of course, Corollary~\ref{cor:kill} is still very far from showing the desired kill number bound $\poly{k}$ as $m$ can be as large as $2^{k-1}$,
but it is hoped that further investigations along this approach may lead to more interesting results.


\subsection{Organization}
The rest of the paper is organized as follows. Preliminaries
and notations that we use throughout the paper are summarized in Section~\ref{sec:prelim}.
We prove our Main Lemma, which deals with the cases when $k$ is at least $5$ in Section~\ref{sec:main_lemma},
while the small value cases are discussed in Section~\ref{sec:small_values}.
Then, by combining these two ingredients, we prove our Main Theorem in Section~\ref{sec:main_thm}.
Finally we end with a brief section of conclusions and open questions.


\section{Preliminaries}\label{sec:prelim}

All logarithms in this paper are to the base $2$.
Let $n\geq 1$ be a natural number, then $[n]$ denotes the set $\{1,\ldots, n\}$.
We use $\F_2$ for the field with $2$ elements $\{0,1\}$, where addition and multiplication are performed modulo $2$.
We view elements in $\F_{2}^{n}$ as $n$-bit binary strings, i.e. elements in $\{0,1\}^n$, interchangeably.
If $x$ and $y$ are two $n$-bit strings, then $x+y$ (or $x-y$) denotes bitwise addition (i.e. XOR) of $x$ and $y$.
For positive integers $m$ and $n$, if $y \in \F_{2}^{m}$ and $z \in \F_{2}^{n}$, then we write
$x=(y, z)$ to denote the binary string $x \in \F_{2}^{m+n}$ obtained from concatenating $y$ and $z$ together.
We view $\F_2^n$ as a vector space equipped with an inner product $ \langle x, y \rangle $,
which we take to be the standard dot product: $\langle x, y \rangle = \sum_{i=1}^n x_iy_i$,
where all operations are performed in $\F_2$.


\subsection{Boolean functions and Fourier analysis}
We often use $f$ to denote a real function defined on $\F_2^n$ and
write $\supp{f}=\{x\in \F_2^n \mid f(x)\neq 0\}$ for the \emph{support} of $f$.
Sometimes we view $f$ as a $2^n$-dimensional vector, e.g. write $f=\mathbf{0}$ and $f=\mathbf{1}$
to denote the trivial all-zero function and all-one function, respectively.
In this paper, a function $f$ is \emph{Boolean} if its range is $\{0,1\}$.

For every $\alpha \in \F_{2}^{n}$, one can define a \emph{linear function} (or \emph{parity function})
mapping $\F_{2}^{n}$ to $\{0,1\}$ as $\ell_{\alpha}(x)=\langle \alpha, x \rangle$.
Let $\chi_{\alpha}=(-1)^{\ell_{\alpha}}$, which are commonly known as \emph{characters}.
For functions $f, g\colon \F_2^n \to \R$ the inner product is defined as
$\langle f, g\rangle := \Exp_{x\in \F_2^n} (f(x) g(x))$. For $\alpha = (\alpha_1, \ldots, \alpha_n) \in \F_2^n$, the corresponding
character function $\chi_\alpha$ is defined as $\chi_\alpha(x_1, \ldots, x_n) = \prod_{i\colon \alpha_i = 1} (-1)^{x_i} = (-1)^{\langle \alpha,x \rangle}$.
For $\alpha,\beta \in \F_{2}^{n}$, the inner product between $\chi_\alpha$ and $\chi_\beta$
is 1 if $\alpha=\beta$, and $0$ otherwise.
Therefore the characters form an orthonormal basis for real-valued functions over $\F_{2}^{n}$,
and we can expand any $f$ defined on $\F_{2}^{n}$ using $\{\chi_{\alpha}\}_{\alpha \in \F_{2}^{n}}$ as a basis.
\begin{definition}[Fourier Transform]
Let $f\colon \F_{2}^{n} \to \R$.
The \emph{Fourier transform} $\hat{f}\colon \F_{2}^{n} \to \C$ of $f$ is defined to be
 $\hat{f}(\alpha)=\Exp_{x}(f(x)\chi_{\alpha}(x))$.
The quantity $\hat{f}(\alpha)$
is called the \emph{Fourier coefficient} of $f$ at $\alpha$.
\end{definition}
The Fourier inversion
formula is given by $f(x)=\sum_{\alpha\in\F_{2}^{n}}\hat{f}(\alpha)\chi_{\alpha}(x)$,
and the Parseval's identity is
$\sum_{\alpha\in \F_{2}^{n}}\hat{f}(\alpha)^{2}=\Exp_{x}(f(x)^{2})$.
The Fourier sparsity of $f$, denoted by $\|\hat f\|_0$,
is the number of nonzero Fourier coefficients of $f$.

\subsubsection{Fourier characterization of Boolean functions}
Our proof crucially relies on the following characterization of Boolean functions in terms of their Fourier spectra.
We give a proof for completeness.
\begin{proposition}[Folklore]\label{prop:Boolean}
A function $f:\F_2^n \to \R$ defined on the hypercube is Boolean if and only if for every $\alpha \in \F_2^n$,
\begin{equation}\label{eqn:Boolean}
\hat{f}(\alpha)=\sum_{\beta\in \F_2^n}\hat{f}(\beta)\hat{f}(\alpha+\beta).
\end{equation}
\end{proposition}
\begin{proof}
This follows from the fact that $f$ is Boolean if and only if $f^2(x)-f(x)=0$ for every $x$. Now expand
the left-hand side in terms of Fourier coefficients and notice that, since the right-hand side is the $\mathbf{0}$-function,
all of its Fourier coefficients all zero.
Comparing each pair of the corresponding Fourier coefficients on both sides gives the desired equality.
\end{proof}

\subsubsection{Linear restrictions}
The following is a folklore theorem regarding the effect of linear restrictions on the
Fourier spectrum of a function defined over the Boolean hypercube.
We include a proof in Appendix~\ref{sec:restriction} for completeness.
\begin{proposition}\label{prop:restriction}
Let $f: \F_2^n \to \R$ be a function defined on the Boolean hypercube.
Let $f_0, f_1: \F_{2}^{n-1} \to \R$ be the ``sub-functions'' obtained from restricting the first bit of the input to $0$ and $1$, respectively;
that is, $f_0(y):= f(0,y)$ and $f_1(y) := f(1, y)$ for all $y\in \F_{2}^{n-1}$.
Then the Fourier spectra of $f_0$ and $f_1$ satisfy that, for all $\beta\in \F_{2}^{n-1}$,
\begin{align}\label{eqn:subfunction}
\hat{f}_{0}(\beta) = \hat{f}(0,\beta) + \hat{f}(1,\beta),
\qquad
\hat{f}_{1}(\beta) = \hat{f}(0,\beta) - \hat{f}(1,\beta).
\end{align}
Conversely, the Fourier spectrum of $f$ satisfies
\begin{align}\label{eqn:original_function}
\hat{f}(0,\beta) = \frac{1}{2} (\hat{f}_{0}(\beta) + \hat{f}_{1}(\beta)),
\qquad
\hat{f}(1,\beta) = \frac{1}{2} (\hat{f}_{0}(\beta) - \hat{f}_{1}(\beta)).
\end{align}
\end{proposition}

\subsubsection{Tensor product}
The statement as well as the proof of Main Theorem requires the standard notion of tensor products between functions.
\begin{definition}[Tensor Product of Boolean Functions]
Let $f:\F_2^{n_{1}}\to \{0,1\}$ and $g:\F_2^{n_{2}}\to \{0,1\}$
be two Boolean functions on $n_{1}$ and $n_{2}$ variables respectively.
Then the tensor product of $f$ and $g$, denoted by $f \otimes g$, is a Boolean function
over $\F_2^{n_{1}+n_{2}}$ such that
$f \otimes g(x, y)=f(x)\cdot g(y)$
for all $x\in \F_2^{n_{1}}$ and $y\in \F_2^{n_{2}}$.
\end{definition}
It is easy to verify the following fact.
\begin{fact}
If $h=f \otimes g$ is the tensor product of two Boolean function defined above, then the Fourier spectrum
$h$ satisfies that $\hat{h}(\alpha, \beta)=\hat{f}(\alpha) \cdot \hat{g}(\beta)$,
for every $\alpha\in \F_2^{n_{1}}$ and $\beta \in \F_2^{n_{2}}$.
\end{fact}
Given a Boolean function $f:\F_2^{n_{1}}\to \{0,1\}$, two commonly used functions to tensor with $f$ are the all-one function 
$g_1=\mathbf{1}$ whose Fourier spectrum is $\hat{g}_1(\mathbf{0})=1$ and $\hat{g}_1(\alpha)=0$ for any $\alpha \neq \mathbf{0}$; 
and the ``$\delta$-function'' $g_2$ defined by $g_2(x)=1$ if and only if $x=0^{n_2}$, whose 
Fourier spectrum is $\hat{g}_2(\alpha)=1/2^{n_2}$ for every $\alpha$. 
Note that tensoring $f$ with $g_1$ is equivalent to setting each to the $2^{n_2}$ sub-functions, defined by restricting $y$ 
to different values in $\F_2^{n_{2}}$, to $f$;
and tensoring $f$ with $g_1$ is to set the sub-function with $y=\mathbf{0}$ to $f$ and set all other sub-functions to the all-zero function.

\subsubsection{Invertible linear transformations and linear shifts}
Let $L: \F_2^n \to \F_2^n$ be an invertible linear transformation. If $f:\F_2^n \to \{0,1\}$ is a Boolean function,
then define $g:=Lf$, the function obtained from applying the linear transformation $L$ to $f$, as $g(x)=f(Lx)$
for all $x\in \F_2^n$.
The Fourier spectrum of $g$ is given by $\hat{g}(\alpha)=\hat{f}((L^T)^{-1}\alpha)$, 
where $L^T$ stands for the transpose of $L$ viewed as an $n\times n$ matrix. 
One can check that the set of Fourier coefficients as well as
the property of being the indicator function of an (affine) linear subspace are invariant under invertible linear transformations.
If $a\in \F_2^n$ is a non-zero vector, and let $h(x):=f(x+a)$ be the linear shift of $f$, 
then the Fourier spectrum of $h$ is given by 
$\hat{h}(\alpha)=\chi_{a}(\alpha)\hat{f}(\alpha)$ for every $\alpha \in \F_2^n$.

\subsection{Additive combinatorics}
Additive combinatorics is the sub-field of mathematics concerned with subsets of integers or more generally abelian groups,
and studies the interplay between the structural properties of a subset and its combinatorial estimates associated with arithmetic operations.
Recently additive combinatorics has found many applications in computer science, see the excellent exposition~\cite{Lov17}
and the textbook~\cite{TV06} for comprehensive treatments.

Throughout this paper, $G$ is the abelian group $\F_2^n$ for some positive integer $n$ 
and the underlying field is $\F_2$.
If $A=\{a_1, \ldots, a_m\}\subset G$, then $\vspan{A}$ stands for the \emph{linear span} of $A$:
$\vspan{A}=\{\sum_{i\in S}a_i \mid S\subseteq [m] \}$, where summation over the empty set is understood to be the $0$ element by convention.
For any $x\in G$ and $A \subset G$, we write $x+A$ to denote the set $\{x+a \mid a\in A\}$.
If $A$ and $B$ are two subsets of $G$, then $A+B$ denotes the \emph{sumset} $\{a+b \mid a\in A \text{ and } b\in B\}$.
Similarly, $A-B:=\{a-b \mid a\in A \text{ and } b\in B\}$, although $A-B$ is always the same as $A+B$ in this paper as the underlying ambient group
is $\F_2^n$.
If $A=B$ then we write $2A:=A+A$ and in general write $kA:=\underbrace{A+\cdots+A}_{k \text{ times}}$ for integer $k\geq 1$.

The following Lemma of Laba is useful for our proofs.
\begin{lemma}[\cite{Lab01}, Theorem 2.5]\label{lemma:Laba}
Let $G$ be an abelian group and $A \subset G$ be a subset of $G$ such that $|A-A|<\frac{3}{2}|A|$. Then $A-A$ is a subgroup of $G$.
\end{lemma}

\section{Proof of the Main Lemma}\label{sec:main_lemma}
First recall our Main Lemma states the following.
\begin{lemma-main}
Let $k\geq 5$ and $n\geq k$ be integers.
Let $f: \F_2^n \to \{0,1\}$ be a Boolean function such that $\hat{f}(\mathbf{0})=1/2^{k-1}$ and any other Fourier coefficients are either zero
or equal to $\pm \frac{1}{2^k}$, then $f$ is the indicator function of a disjoint union of two dimension $n-k$
affine subspaces.
\end{lemma-main}

In Appendix~\ref{sec:Fourier_spectrum}, we compute the Fourier spectrum of a Boolean function that is supported on
two disjoint affine subspaces such that the two affine subspaces are of the same dimension
and their Fourier spectra have minimum intersection.
Our strategy for the proof of the Main Lemma is to show that if the Fourier coefficients of a Boolean function satisfy
the condition prescribed in the Main Lemma, then its Fourier spectrum matches the one we show in Appendix~\ref{sec:Fourier_spectrum}.

Let us define
\[
A=\{\alpha \in \F_2^n \mid \hat{f}(\alpha)=\frac{1}{2^k}\}
\]
and
\[
B=\{\beta \in \F_2^n \mid \hat{f}(\beta)=-\frac{1}{2^k}\}.
\]

Without loss of generality\footnote{This is because if $f(\mathbf{0})=0$, 
then let $a\in \F_2^n$ be any vector such that $f(a)=1$.
We can apply a linear shift $a$ to $f$ to get a new Boolean function, 
$h(x)=f(x+a)$ for every $x$, so that $h(\mathbf{0})=1$. 
Note that the conclusions in our Main Theorem are invariant under linear shifts.
Moreover,
since $\hat{h}(\alpha)=\chi_{a}(\alpha)\hat{f}(\alpha)$ for every $\alpha \in \F_2^n$, we have
$\hat{h}(\mathbf{0})=1/2^{k-1}$ and $|\hat{h}(\alpha)|=|\hat{f}(\alpha)|$ for any other nonzero $\alpha$.
Therefore, the assumptions apply to $h$ as well.}, from now on, we may assume $f(\mathbf{0})=1$.
We begin with calculating the cardinalities of sets $A$ and $B$.

\begin{claim}\label{claim:sizes}
For any $k\geq 1$ and $n\geq k$, we have $|A|=3t$ and $|B|=t$, where $t=2^{k-1}-1$.
\end{claim}
\begin{proof}
Since $\hat{f}(\mathbf{0})=1/2^{k-1}$, by Parseval's identity
$\hat{f}(\mathbf{0})=1/2^{k-1}=\sum_{\alpha \in \F_2^n}\hat{f}^2(\alpha)$, we have $|A|+|B|=2^{k+1}-4$.

On the other hand,
\[
1=f(\mathbf{0})=\sum_{\alpha \in \F_2^n}\hat{f}(\alpha)\chi_{\alpha}(\mathbf{0})
=\frac{1}{2^{k-1}}+\sum_{\alpha \in A}\frac{1}{2^k} + \sum_{\beta \in B}(-\frac{1}{2^k}),
\]
which gives $|A|-|B|=2^k-2$. Therefore we have $|A|=3(2^{k-1}-1)$ and $|B|=2^{k-1}-1$.
\end{proof}

For convenience, we let $A=\{\alpha_1, \ldots, \alpha_{3t}\}$ and $B=\{\beta_1, \ldots, \beta_t\}$ in the following.

\subsection{Some additive properties of sets \texorpdfstring{$A$}{} and \texorpdfstring{$B$}{}}
We now study the additive properties of sets $A$ and $B$.
Note that the Fourier coefficients of $f$ are non-zero only at $\mathbf{0}$ and in sets $A$ and $B$;
moreover, the Fourier coefficients are uniform for points in $A$ or $B$. Therefore,
by Proposition~\ref{prop:Boolean}, we expect that there are nice additive structures within $A$ and $B$.

\begin{definition}
We call $(\alpha, \beta, \alpha +\beta)$ a \emph{triangle} if $\alpha$, $\beta$ and $\alpha +\beta$
are all in the support of $\hat{f}$; that is $\alpha, \beta, \alpha +\beta \in A\cup B \cup \{\mathbf{0}\}$.
\end{definition}

\begin{lemma}\label{lemma:beta_triangle}
For any $\beta_i \in B$, there are exactly $t$ triangles passing through $\beta_i$;
namely, the $t$ triangles are $(\beta_i, \beta_i, \mathbf{0})$ and 
$\{(\beta_i, \beta_j, \beta_i + \beta_j)\}_{j=1, j\neq i}^{t}$.
In the language of set addition, we have $2B \subseteq A \cup \{\mathbf{0}\}$.
\end{lemma}
\begin{proof}
For any $\beta_i \in B$, by Proposition~\ref{prop:Boolean},
\begin{align*}
\hat{f}(\beta_i) 
&= -\frac{1}{2^k} =\sum_{\gamma \in \F_2^n}\hat{f}(\gamma)\hat{f}(\beta_i + \gamma) \\
&=2\hat{f}(\mathbf{0})\hat{f}(\beta_i)+ \sum_{\substack{j=1 \\ j\neq i}}^{t}\hat{f}(\beta_j)\hat{f}(\beta_i+\beta_j) +
  \sum_{\ell=1}^{3t}\hat{f}(\alpha_\ell)\hat{f}(\beta_i+\alpha_\ell) \\
&\geq 2\cdot \frac{1}{2^{k-1}}\cdot (-\frac{1}{2^k})+2(t-1)(-\frac{1}{2^k})(\frac{1}{2^k}) \text{ \footnotemark} \\
&=-\frac{1}{2^k},
\end{align*}
\footnotetext{There is a factor $2$ in the second summation because if $\beta_i+\beta_j \in A$, then
the triangle $(\beta_i, \beta_j, \beta_i+\beta_j)$ appears twice in the summation
$\sum_{\gamma \in \F_2^n}\hat{f}(\gamma)\hat{f}(\beta_i + \gamma)$:
once with $\gamma=\beta_j$ and the other with $\gamma=\beta_i+\beta_j$.}where the inequality in the second last line 
becomes equality if and only if the following two conditions hold:
1) for every $1\leq j \leq t$, $j\neq i$, $\beta_i+\beta_j \in A$; and
2) there is no triangle of the form $(\beta_i, \alpha_j, \alpha_\ell)$. Hence the lemma follows.
\end{proof}

\begin{corollary}\label{cor:sum-free}
The set $B$ is a sum-free set; namely, for any three elements $\beta_1, \beta_2, \beta_3 \in B$, $\beta_1+\beta_2\neq \beta_3$.
Equivalently, $2B \cap B = \emptyset$.
\end{corollary}
\begin{proof}
This follows directly from Lemma~\ref{lemma:beta_triangle} and the fact sets $A$ and $B$ are disjoint.
\end{proof}

\begin{corollary}\label{cor:2B-3B}
We have $2B \cap 3B = \emptyset$.
\end{corollary}
\begin{proof}
Suppose not, then there exist $\beta_1, \beta_2, \beta_3, \beta_4, \beta_5$ in $B$ such that
$\beta_1 + \beta_2 = \beta_3 + \beta_4 + \beta_5$. These five elements must be distinct as otherwise
they would give rise to a triangle in $B$. But then we have a $(\alpha_1, \alpha_2, \beta_5)$ triangle,
where $\alpha_1:=\beta_1 + \beta_2$ and $\alpha_2:=\beta_3 + \beta_4$, contradicting to Lemma~\ref{lemma:beta_triangle}.
\end{proof}

Let us define
\[
R=2B\cap A = 2B\setminus \{\mathbf{0}\}
\]
and
\[
L=A\setminus R.
\]
Note that $L$ and $R$ are disjoint and $A=L \cup R$.
For any $\rho \in R$, let
\[
N(\rho)=\{\beta_i \in B \mid \text{$\exists \beta_j \in B$ s.t. $\rho=\beta_i + \beta_j$ }\}
\]
be the set of points in $B$ which has a triangle passing through $\rho$. Define a set $\Gamma\subset \F_2^n$ as
\[
\Gamma=\{\gamma = \rho + \beta \mid \rho\in R, \beta \in B \text{ and } \beta \notin N(\rho)\}.
\]
Observe that $\Gamma$ is nonempty: since for every $\rho\in R$, all its $\beta$-neighbors can be paired together,
so $|N(\rho)|$ is an even number, but $|B|=2^{k-1}-1$ is odd.

\begin{claim}\label{claim:Gamma}
We have $\Gamma = 3B \setminus B$.
\end{claim}
\begin{proof}
On one hand, by the definition of set $\Gamma$, $\Gamma \subseteq 3B$; since $R$ and $B$ are disjoint and $\mathbf{0}\notin R$,
we have $\Gamma \cap B =\emptyset$, and hence $\Gamma \subseteq 3B \setminus B$.
On the other hand, let $\gamma$ be any element in $3B$; that is $\gamma=\beta_1+\beta_2+\beta_3$,
where $\beta_1,\beta_2,\beta_3\in B$.
When will $\gamma$ actually be in $B$?
This happens only if any two of these three elements are identical, 
then $\gamma=\beta_i$ for some $i \in \{1,2,3\}$,
thus $\gamma\in B$.
Moreover, assume that these three elements are distinct and suppose $\gamma \in B$, i.e. $\gamma=\beta_j$ for some $j>3$.
Let $\rho:=\beta_1+\beta_2$, then $\rho=\beta_3+\gamma=\beta_3+\beta_j$;
that is $\gamma=\rho+\beta_3$ and $\beta_3 \in N(\rho)$.
Therefore, if $\gamma \in 3B \setminus B$, then we must have $\beta_3 \notin N(\rho)$
and consequently $\gamma \in \Gamma$. It follows that $3B \setminus B \subseteq \Gamma$.
This completes the proof of the claim.
\end{proof}

It is easy to see that $\Gamma$ is disjoint from the Fourier support of $f$.
\begin{claim}\label{claim:Gamma_fourier_support}
For every element $\gamma \in \Gamma$, we have $\hat{f}(\gamma)=0$.
\end{claim}
\begin{proof}
Recall that, the support of $\hat{f}$ is $A \cup B \cup\{\mathbf{0}\}$.
Suppose $\hat{f}(\gamma) \neq0$, that is $\gamma \in \supp{\hat{f}}$.
Since $A=L \cup R$, from Claim~\ref{claim:Gamma},
we know that $\Gamma \cap B =\emptyset$; from Claim~\ref{claim:Gamma} and Corollary~\ref{cor:2B-3B},
we know that $\Gamma \cap 2B = \Gamma \cap (R \cup \{\mathbf{0}\}) =\emptyset$.
So there is only one possibility left, which is $\gamma \in L$.
However, if this were the case, because $\gamma=\rho+\beta$ with $\rho\in R$, 
it would give rise to a $(\gamma, \rho, \beta)$-triangle
with $\gamma, \rho \in A$, contradicting Lemma~\ref{lemma:beta_triangle}, 
so $\gamma$ is not in $L$, hence $\hat{f}(\gamma)=0$.
\end{proof}


\subsection{Even-Zohar's tight bound on \texorpdfstring{$F(K)$}{}}
Let $G$ be an abelian group and $A\subset G$ be a subset. The fundamental Freiman theorem~\cite{Fre73} in additive combinatorics states that
if $G$ is $\Z$ and $|A+A|\leq K|A|$ for some constant $K$, then there exist functions $d(K)$ and $\ell(K)$ such that
$A$ is contained in a $d(K)$-dimensional arithmetic progression of length at most $\ell(K)|A|$.
The ratio $\sigma[A]:=|A+A|/|A|$ is commonly known as the \emph{doubling constant} of set $A$.
Hence Freiman theorem asserts that if a set of integers has small doubling constant, then the set is well-structured.
Ruzsa~\cite{Ruz99} established an analog of Freiman's theorem for finite abelian groups with torsion $r$.
Specifically, he proved that any subset $A$ with doubling constant $K$ is contained in a subgroup of $G$ of size at most $K^2 r^{K^4}|A|$.
The question for groups $\F_2^n$ was first studied by Green and Ruzsa~\cite{GR06} and the bound was later improved by Sanders~\cite{San08}.
An asymptotically tight bound was first proved in~\cite{GT09b} and \cite{Kon08}.

For a subset $A\subset \F_2^n$, let $\langle A \rangle$ denote the \emph{affine span} of $A$;
namely, the smallest affine subspace that contains $A$.
If $\sigma[A]=K$, then let $F(K):=\max_{A: \sigma[A]=K}|\langle A \rangle|/|A|$ denote the maximum relative size of the affine span of $A$.
Even-Zohar~\cite{Eve12} gave the tight bound of $F(K)$ for all values of doubling constant $K$.

\begin{theorem}[\cite{Eve12}, Theorem 2]\label{theorem:Even-Zohar}
Let $A$ be a subset of $\F_2^n$ with doubling constant $K$, i.e. $|2A|/|A|\leq K$.
If $s$ is the unique positive integer satisfying the inequalities
\begin{equation}\label{eqn:K}
\frac{\binom{s}{2}+s+1}{s+1} \leq K < \frac{\binom{s+1}{2}+s+2}{s+2},
\end{equation}
then $|\langle A \rangle|/|A| \leq F(K)$, where $F(K)$ is given by
\begin{align}\label{eqn:F(K)}
F(K)=
\begin{cases}
\frac{2^s}{\binom{s}{2}+s+1} \cdot K 	& \text{ if $\frac{\binom{s}{2}+s+1}{s+1} \leq K < \frac{s^2+s+1}{2s}$,} \\
\frac{2^{s+1}}{s^2+s+1} \cdot K			& \text{ if $ \frac{s^2+s+1}{2s} \leq K < \frac{\binom{s+1}{2}+s+2}{s+2}$.}
\end{cases}
\end{align}
\end{theorem}

\subsection{Characterizing \texorpdfstring{$2B$}{} and \texorpdfstring{$\vspan{B}$}{}}\label{sec:property}

Note that the doubling constant of set $B$ satisfies that
\begin{align} \label{eqn:sigma_B}
\sigma[B]=\frac{|R|+1}{|B|} \leq \frac{|A|+1}{|B|}=3+\frac{1}{t},
\end{align}
and recall that $t=2^{k-1}-1$. Therefore, when $k\geq 5$, $K=\sigma[B]\leq \frac{46}{15}$. Plugging this $K$ into
\eqref{eqn:K} gives that $s \leq 5$ and consequently $F(K)\leq 2K<7$. That is, we have
$|\langle B \rangle|< 7|B|$.

The most important step in our proof is establishing the following lemma, 
which almost completely characterizes the structure of set $B$.
\begin{lemma}\label{lemma:2B}
If $k\geq 5$, then $|\vspan{B}|=2^k=2(|B|+1)$ and $2B$ is a subspace of dimension $k-1$.
\end{lemma}
We prove Lemma~\ref{lemma:2B} in the following two subsections, distinguishing between
the case when $\langle B \rangle$ is an affine subspace
and the case when $\langle B \rangle$ is a subspace.

\subsubsection{If \texorpdfstring{$\langle B \rangle$}{} is an affine subspace}

In the case that $\langle B \rangle$ is an affine subspace, let $\langle B \rangle=a+H$ be the affine subspace,
where $H$ is a subspace of $\F_2^n$, $a\in H^\perp$ and $a\neq \mathbf{0}$.
Therefore $\vspan{B}=H \cup (a+H)$. Note that we now have
$2\ell B\subseteq H$ and $(2\ell -1)B\subseteq a+H$ for every integer $\ell \geq 1$.
Moreover, $|\vspan{B}|=2|\langle B \rangle|<14|B|$. Since $\vspan{B}$ is a subspace and $|B|=2^{k-1}-1$,
so there are only three possibilities:
$|\vspan{B}|=8(|B|+1)$, $|\vspan{B}|=4(|B|+1)$ and $|\vspan{B}|=2(|B|+1)$. In the following, we are going to eliminate
the first two possibilities.

\begin{claim}
Set $L$ is nonempty.
\end{claim}
\begin{proof}
Suppose not, then $2B=A \cup \{\mathbf{0}\}\subset H$. 
Recall that by Claim~\ref{claim:Gamma}, $\Gamma=3B\setminus B$,
so $\Gamma \subseteq a+H$ and is disjoint from set $A$. 
It follows that for any $\gamma\in \Gamma$, $\hat{f}(\gamma)=0$ (or directly from Claim~\ref{claim:Gamma_fourier_support}).
However, applying Proposition~\ref{prop:Boolean} to $\hat{f}(\gamma)$, we see that by the definition of set $\Gamma$,
$\gamma=\rho+\beta$ with $\rho\in A$, $\beta\in B$ and $\beta \notin N(\rho)$. Hence there is at least one negative term contribution
on the right-hand side in~\eqref{eqn:Boolean} for $\hat{f}(\gamma)$, 
but since both $2B$ and $2A$ are disjoint from $\Gamma$, there is no positive term
on the right-hand side in~\eqref{eqn:Boolean}, a contradiction.
\end{proof}

We discuss the following two possibilities separately.

\medskip 
\paragraph{The case when $|H|=4(|B|+1)$.}
First note that if this were the case, then $F(K)=|\langle B \rangle|/|B|=4(1+\frac{1}{|B|})$.
By Theorem~\ref{theorem:Even-Zohar}, the doubling constant of $B$ is at least $K=|2B|/|B|>2.5$, or $|2B|>2.5|B|$.
Therefore $|L|\leq 0.5|B|$.
On the other hand, $4B=2B+2B$ and $4B\subseteq H$ so $\sigma[2B]=|4B|/|2B|<\frac{4+\frac{1}{16}}{2.5}<7/4$. Then by
Theorem~\ref{theorem:Even-Zohar} again, $|4B|=|\langle 2B \rangle|$, that is $4B=H$.

We next claim that $L \subseteq H$. To see this, let $\lambda$ be an arbitrary element in $L$; applying
Proposition~\ref{prop:Boolean} to $\hat{f}(\lambda)$ gives
\[
\frac{1}{2^k}=\hat{f}(\lambda)=2\hat{f}(\lambda)\hat{f}(\mathbf{0})+
\sum_{\lambda'\in L}\hat{f}(\lambda')\hat{f}(\lambda+\lambda')+\text{other terms}.
\]
The first term and the second summation can contribute at most $\frac{1}{2^{2k}}(2|L|+2)\leq \frac{|B|+2}{2^{2k}}<\frac{1}{2^{k}}$.
Therefore, the ``other terms'' on the right-hand side must contain terms of the form $\hat{f}(\alpha_1)\hat{f}(\alpha_2)$,
where $\alpha_1$ and $\alpha_2$ are two distinct points in $A$ and $\lambda=\alpha_1+\alpha_2$.
That is $\lambda \in 2B+2B$, hence it follows that $L\subseteq 4B=H$.

Let $D:=H\setminus (2B\cup L)$. We have $|D|=4(|B|+1)-3|B|-1=|B|+3>0$. Let $\delta$ be any point in $D$. First,
since $\delta \notin 2B \cup B$, $\hat{f}(\delta)=0$. Second, since $\delta \in H$, there is no negative term
in the right-hand side of $0=\hat{f}(\delta)=\sum_{\gamma \in \F_2^n}\hat{f}(\gamma)\hat{f}(\delta+\gamma)$,
because if $\gamma \in B$, then $\delta+\gamma \in a+H$ but there is no positive Fourier coefficient in $a+H$ (since $L\subset H$).
On the other hand, consider the set $\{\delta+\alpha \mid \alpha \in 2B \cup L\}$. Since $|D|<|H|/2$, this set has non-empty
intersection with $2B\cup L$. Therefore, there are positive terms in $\sum_{\gamma \in \F_2^n}\hat{f}(\gamma)\hat{f}(\delta+\gamma)$,
this contradicts the fact that $\hat{f}(\delta)=0$.

\medskip 
\paragraph{The case when $|H|=2(|B|+1)$.}
This case is similar to the previous one. First, if this were the case,
then $F(K)=|\langle B \rangle|/|B|=2(1+\frac{1}{|B|})$. It follows that,
by Theorem~\ref{theorem:Even-Zohar}, the doubling constant of $B$ is at least $K=|2B|/|B|>7/4$, and hence
$|4B|/|2B|\leq |H|/|2B|<3/2$, and by Theorem~\ref{theorem:Even-Zohar} again $4B=H$.
The rest is identical to the case when $|H|=4(|B|+1)$.

\begin{proof}[Proof of Lemma~\ref{lemma:2B} when $\langle B \rangle$ is an affine subspace]
Now that the only possibility left is $|\vspan{B}|=2 \cdot (|B|+1)$, and because $\langle B \rangle$ is an affine subspace,
it follows that $2B \subseteq H$ and hence $|2B|\leq |B|+1$. Applying Laba's lemma, Lemma~\ref{lemma:Laba}, to set $B$
gives that $2B$ is a subspace. Since $|2B|\geq |B|$, it follows that $2B=H$, a dimension $k-1$ subspace.
\end{proof}

\subsubsection{If \texorpdfstring{$\langle B \rangle$}{} is a subspace}

If the affine span $\langle B \rangle$ is a subspace, and since $|\langle B \rangle|<7|B|$, then we either have
$|\langle B \rangle|=4(|B|+1)$ or $|\langle B \rangle|=2(|B|+1)$ (because $B \cap 2B = \emptyset$ and $|2B|\geq |B|$,
 $|\langle B \rangle|\geq 2|B|$). In the following we exclude the first case.

Recall that $R=2B\setminus \{\mathbf{0}\}$ is the set of non-zero points in the Fourier support of $f$ that can be written as a sum of
two $\beta$-points in $B$. Let $R=\{\lambda_1, \ldots, \lambda_m\}$, where $m$ is the cardinality of $R$.

\begin{claim}
If $\langle B \rangle$ is a subspace, then $m \leq 2.5t$.
\end{claim}
\begin{proof}
For the sake of contradiction, suppose that $m > 2.5t$.
For every $\lambda_i \in R$, let $d_i$ be the number of $\beta_j$'s that form a triangle with $\lambda_i$.
Then we have $\sum_{i=1}^{m}d_i=t(t-1)$ and $d_i\geq 2$ for every $1\leq i \leq m$.
By a standard averaging argument, there is some $\lambda_i$ with $d_i\leq 0.4t$.
By the definition of set $\Gamma$, it follows that $|\Gamma|\geq t-d_i =0.6t$.
Recall that $\Gamma=3B\setminus B$ so $\Gamma \subset \langle B \rangle=\vspan{B}$, and $\Gamma$ is disjoint from either $2B$ or $B$,
thus $|\langle B \rangle|\geq |2B|+|B|+|\Gamma|>4.1t$, contradicting our assumption that $|\langle B \rangle|=4(|B|+1)$.
\end{proof}

\begin{proof}[Proof of Lemma~\ref{lemma:2B} when $\langle B \rangle$ is a subspace]
Now since $m\leq 2.5t$, the doubling constant of $B$ is at most $|2B|/|B|\leq 2.5+1/|B|<21/8$, then by Theorem~\ref{theorem:Even-Zohar},
$|\langle B \rangle|/|B| <42/11< 4$, therefore we must have $|\langle B \rangle|=2(|B|+1)=2^k$.
Once again, applying Laba's lemma to set $B$ shows that $2B$ is a subspace of dimension $k-1$.
\end{proof}

\subsection{Completing the proof of the Main Lemma}\label{sec:remain}
By Lemma~\ref{lemma:2B}, $2B$ is a dimension $k-1$ subspace; without loss of generality, we may assume that
\begin{equation}\label{eqn:2B}
H=2B=\vspan{e_1, \ldots, e_{k-1}}.
\end{equation}
Since $|\vspan{B}|=2^k=2|2B|$, and $B\cap 2B=\emptyset$, $B$ is an affine shift of $H$ with one point $\delta$ missing.
Since $\delta \notin H$, so without loss of generality, we may assume $e_k$ is the missing point. That is
\begin{align}
 B &= (e_k + \vspan{e_1, \ldots, e_{k-1}}) \setminus \{e_k\} \quad \text{and} \label{eqn:B} \\
 R &= 2B \setminus \{\mathbf{0}\} = \vspan{e_1, \ldots, e_{k-1}} \setminus \{\mathbf{0}\} = e_k+B. \label{eqn:R}
\end{align}

Now by Claim~\ref{claim:Gamma}, we have $\Gamma=\{e_k\}$ and consequently $\hat{f}(e_k)=0$. Our last task is to determine
the structure of set $L$. Recall that $A=R \cup L$ and $|A|=3t$, and because we now have $R=2B\setminus \{\mathbf{0}\}$,
therefore $|L|=2t=2^k-2$.

\begin{claim}\label{claim:L1}
For any $\lambda \in L$, $e_k+\lambda \in L$.
\end{claim}
\begin{proof}
Applying Proposition~\ref{prop:Boolean} to the Fourier coefficient of $f$ at $e_k$ and noting that
$R=e_k+B$, we have
\begin{align*}
\hat{f}(e_k)
&=0=\sum_{\gamma \in \F_2^n}\hat{f}(\gamma)\hat{f}(e_k +\gamma) \\
&=2\sum_{\rho \in R}\hat{f}(\rho)\hat{f}(e_k +\rho) + \sum_{\lambda \in L}\hat{f}(\lambda)\hat{f}(e_k +\lambda) \\
&\leq 2t\cdot(-\frac{1}{2^{2k}}) + 2t \cdot \frac{1}{2^{2k}} \\
&=0,
\end{align*}
where equality holds in the second last line only if for every $\lambda \in L$, $\hat{f}(e_k+\lambda)=\frac{1}{2^k}$.
That is, $e_k+\lambda \in A (=L\cup R)$. As each element in $R$ has already been taken into account in
the first summation in the second line, therefore we necessarily have $e_k+\lambda \in L$.
\end{proof}

\begin{claim}\label{claim:L2}
For any $\lambda \in L$ and $\rho \in R$, $\hat{f}(\lambda+\rho)=0$.
\end{claim}
\begin{proof}
Applying Proposition~\ref{prop:Boolean} to $\hat{f}(\rho)$, where $\rho$ is an arbitrary element in $R$, we have
\begin{align*}
\hat{f}(\rho)
&=\frac{1}{2^k}=2\cdot \hat{f}(\mathbf{0})\hat{f}(\rho)+\sum_{\beta \in B}\hat{f}(\beta)\hat{f}(\rho +\beta)
   + \sum_{\rho' \in R, \rho' \neq \rho}\hat{f}(\rho')\hat{f}(\rho +\rho')
   + \sum_{\lambda \in L}\hat{f}(\lambda)\hat{f}(\lambda+\rho) \\
&=2\cdot \frac{2}{2^k}\cdot \frac{1}{2^k} + (t-1)\cdot (-\frac{1}{2^k})\cdot (-\frac{1}{2^k})
   + (t-1)\cdot (\frac{1}{2^k})\cdot (\frac{1}{2^k})+ \sum_{\lambda \in L}\frac{1}{2^k} \cdot \hat{f}(\lambda+\rho)\\
&\geq \frac{1}{2^k}, \quad \quad \text{(as $\lambda +\rho \notin B$, therefore $\hat{f}(\lambda+\rho)\geq 0$)}
\end{align*}
where we have a factor of $(t-1)$ in the second line because $\rho+e_k \in B$ and
equality holds in the last line only if $\hat{f}(\lambda+\rho)=0$ for every $\lambda \in L$ and every $\rho \in R$.
\end{proof}

\begin{claim}\label{claim:L3}
For any $\lambda, \lambda' \in L$, $\lambda + \lambda' \in L$ except that $\lambda + \lambda' = \mathbf{0}\text{ or }e_k$.
\end{claim}
\begin{proof}
Applying Proposition~\ref{prop:Boolean} to $\hat{f}(\lambda)$, where $\lambda$ is an arbitrary element in $L$, we have
\begin{align*}
\hat{f}(\lambda)=\frac{1}{2^k}
&=2\cdot \hat{f}(\mathbf{0})\hat{f}(\lambda)+\sum_{\beta \in B}\hat{f}(\beta)\hat{f}(\lambda +\beta)   + \sum_{\rho \in R}\hat{f}(\rho)\hat{f}(\lambda +\rho)  + \sum_{\lambda' \in L, \lambda'+\lambda \notin\{\mathbf{0},e_k\}}\hat{f}(\lambda')\hat{f}(\lambda+\lambda') \\
&=2\cdot \frac{2}{2^k}\cdot \frac{1}{2^k} + 0
   + 0 + \sum_{\lambda' \in L, \lambda'+\lambda \notin\{\mathbf{0},e_k\}}\frac{1}{2^k}\cdot\hat{f}(\lambda+\lambda') \text{ \footnotemark}\\
&\leq \frac{4}{2^{2k}}+(2t-2)\cdot (\frac{1}{2^k})\cdot (\frac{1}{2^k}) \\
&=\frac{1}{2^k},
\end{align*}
\footnotetext{The second term vanishes because the only triangles passing through a point $\beta_i \in B$ are
of the type $(\beta_i, \beta_j, \rho_\ell)$ where $\rho_\ell \in R$;
the third term vanishes because of Claim~\ref{claim:L2}.}
where equality holds in the second last line only if $\lambda+\lambda' \in L$ for every $\lambda' \in L$,
except when $\lambda'$ is equal to $\lambda$ or $\lambda+e_k$.
\end{proof}

Put Claim~\ref{claim:L1}, Claim~\ref{claim:L2} and Claim~\ref{claim:L3} together, and since $|L|=2^k-2$
we conclude that $H':=L \cup \{\mathbf{0}, e_k\}$ is a subspace of dimension $k$.
Moreover, as $\vspan{B}=\vspan{e_1, \ldots, e_k}$ is a subspace of dimension $k$, and
$L\cap \vspan{B}=\emptyset$, we thus have $H' \cap \vspan{B}=\{\mathbf{0},e_k\}$.
Therefore, without loss of generality, we may take $H'=\vspan{e_k, \ldots, e_{2k-1}}$ and consequently finally have
\begin{equation}\label{eqn:L}
L=\vspan{e_k, \ldots, e_{2k-1}} \setminus \{\mathbf{0},e_k\}.
\end{equation}

It is straightforward to check\footnote{The second line in~\eqref{eqn:Fourier_spectrum} corresponds to set $B$,
third line in~\eqref{eqn:Fourier_spectrum} corresponds to set $R$,
and the fourth and fifth lines of~\eqref{eqn:Fourier_spectrum} correspond to set $L$.}
that the Fourier spectrum calculated in Section~\ref{sec:Fourier_spectrum} for
a disjoint union of two dimension $n-k$ affine subspaces
is identical to the Fourier spectrum of $f$, which is completely specified by sets in~\eqref{eqn:B}, \eqref{eqn:R} and \eqref{eqn:L}.
Therefore the proof of the Main Lemma is complete.

\section{Dealing with small values of \texorpdfstring{$k$}{}}\label{sec:small_values}

When $k=2$ or $k=3$, note that since Claim~\ref{claim:sizes} holds for every $k \geq 2$,
this will enable us to prove the same results as Main Lemma by slightly different arguments.
That is, when $k=2$ or $k=3$, support of $f$ is also a disjoint union of two dimension $n-k$ affine subspaces.
However, when $k=4$ one can not prove the same characterization as Main Lemma. 
In fact, there are two possibilities: 
one is that $f$ is still the indicator function of two disjoint dimension $n-4$ affine subspaces; 
the other is that support of $f$ are \emph{four} disjoint $n-5$ affine subspaces. 
Furthermore, we show that this is the only counterexample to Main Lemma for all $k$.
Now we give the precise statements for small values of $k$ and their proofs.

\begin{lemma}\label{lemma:small_values}
Let $2\leq k \leq 4$ and $n\geq k$ be integers.
Let $f: \F_2^n \to \{0,1\}$ be a Boolean function such that $\hat{f}(\mathbf{0})=1/2^{k-1}$ and any other Fourier coefficients are either zero
or equal to $\pm \frac{1}{2^k}$.
If $k=2$ or $k=3$, then $f$ is the indicator function of a disjoint union of two dimension $n-k$ affine subspaces;
If $k=4$, then $f$ is either the indicator function of a disjoint union of two dimension $n-k$ affine subspaces,
or the indicator function of a disjoint union of four dimension $n-k-1$ affine subspaces.
\end{lemma}

\subsection{Proof of the case \texorpdfstring{$k = 2$}{Lg}}
In this case, $|A|=3$ and $|B|=1$. For convenience, suppose that
$\hat{f}(\mathbf{0})=\dfrac{1}{2}$, $\hat f(\beta)=-\dfrac{1}{4}$ and $\hat f(\alpha_1)=\hat f(\alpha_2)=\hat f(\alpha_3)=\dfrac{1}{4}$,
where $\beta, \alpha_1, \alpha_2, \alpha_3$ are four distinct non-zero vectors.

We claim that there exists an $\alpha_i$,  $1\leq i \leq 3$, such that $\hat{f}(\beta+\alpha_i)=0$. 
To see this, suppose $\hat{f}(\beta+\alpha_i)\neq 0$
for every $1\leq i \leq 3$. Because the four vectors are distinct, $\beta+\alpha_i \neq \mathbf{0}$;
furthermore, since $\alpha_i \neq \mathbf{0}$, so $\beta+\alpha_i \neq \beta$. It follows that
$\beta+\{\alpha_1, \alpha_2, \alpha_3\} = \{\alpha_1, \alpha_2, \alpha_3\}$; that is, adding $\beta$ to $A$ permutes the three elements in the set.
But now adding these three elements together gives $3\beta_1+\sum{\alpha_i}=\sum{\alpha_i}$,
a contradiction since $\beta_1 \neq \mathbf{0}$.

Without loss of generality, assume $\hat{f}(\beta+\alpha_1)=0$ and denote $\beta+\alpha_1$ by $\gamma$.
Now applying Proposition~\ref{prop:Boolean} to $\gamma$ gives:
\begin{align*}
\hat f(\gamma) 
& = 0  = \sum\limits_{\alpha} \hat f(\alpha) \hat f(\alpha+\gamma) \\
& = 2\cdot\hat f(\beta_1) \hat f(\alpha_1) + \hat f(\alpha_2) \hat f(\alpha_2+\gamma) +   \hat f(\alpha_3) \hat f(\alpha_3+\gamma)  \\
& = 2\cdot (-\dfrac{1}{4})\cdot \dfrac{1}{4} + \hat f(\alpha_2) \hat f(\alpha_2+\gamma) +   \hat f(\alpha_3) \hat f(\alpha_3+\gamma) \\
& \leq 0,
\end{align*}
where equality holds in the last line only if $\gamma = \alpha_2 + \alpha_3$
so that
\[
\hat{f}(\alpha_2) \hat{f}(\alpha_2+\gamma) = \hat{f}(\alpha_3) \hat{f}(\alpha_3+\gamma) = \hat{f}(\alpha_2) \hat{f}(\alpha_3)
= \dfrac{1}{4} \cdot \dfrac{1}{4}.
\]
After taking an invertible linear transformation if necessary,
we may take $\alpha_1=e_1, \beta=e_1+e_2, \alpha_2=e_3$ and $\alpha_3=e_2+e_3$, 
then it is easy to verify that this is identical to the Fourier spectrum in~\eqref{eqn:Fourier_spectrum} for the case of $k=2$.

\subsection{Proof of the case \texorpdfstring{$k = 3$}{Lg}}
In this case, $|A|=9$ and $|B|=3$. Denote set $B$ by $\{\beta_1,\beta_2,\beta_3\}$.
Then by Corollary~\ref{cor:sum-free}, $\beta_1+\beta_2+\beta_3 \neq 0$, therefore
$R=\{\beta_1+\beta_2,\beta_1+\beta_3,\beta_2+\beta_3\}$, and $\Gamma=\{\beta_1+\beta_2+\beta_3\}$.
Hence Lemma~\ref{lemma:2B} is established and the rest of the proof is identical to that of the Main Lemma in Section~\ref{sec:remain}
for the general $k\geq 5$ case.

\subsection{Proof of the case \texorpdfstring{$k = 4$}{Lg}}
First of all, it is easy to see that when $k=4$, 
the indicator function of a disjoint union of $2$ affine subspaces
of dimension $n-k=n-4$ is still a Boolean function with desired Fourier spectrum, 
for every $n\geq 4$. 
Next we construct another Boolean function, which demonstrates that
Main Lemma is no longer valid for $k=4$.

\begin{construction}\label{construction:counter}
Let $G=\F_2^6$ with $e_1,\cdots,e_6$ as the standard basis and let $A, B \subset G$ be two disjoint subsets given as follows:
\begin{itemize}
\item $B=\{e_i \mid 1\leq i\leq6\}\cup\{\sum_{i=1}^{6}{e_i}\}$;
\item $A=\{e_i+e_j \mid 1\leq i < j \leq6 \} \cup \{\sum_{i\in S}{e_i} \mid S\subset [6], |S|=5\}$.
\end{itemize}
\end{construction}
Clearly $A=2B\setminus \{\mathbf{0}\}$, 
$|B|=2^{4-1}-1=7$ and $|A|=\binom{7}{2}=3|B|$, which satisfy the size requirements for $A$ and $B$ for $k=4$. 
To see that sets $A$ and $B$ in Construction~\ref{construction:counter} satisfy all the additive properties imposed by
Proposition~\ref{prop:Boolean}, one can explicitly compute a ``core'' function $f_{CE}:\F_2^{6}\to \R$
with $A \cup B \cup \{\mathbf{0}\}$ being its Fourier support to verify that $f$ is indeed 
a \emph{Boolean} function and 
$\supp{f_{CE}}= 
\{\mathbf{0}\} 
\cup 
\{\sum_{i\in S}{e_i} \mid S\subset [6], |S|=5\} 
\cup\{\sum_{i=1}^{6}{e_i}\}$.
That is, $f$ is equal to $1$ on vectors of weights $0$, $5$ and $6$, and is equal to $0$ on all other vectors.
Note that $\supp{f_{CE}}$ consists of $8$ distinct vectors and is a disjoint union of four affine subspaces of dimension $n-4-1=1$ each.
Moreover, it can be checked that $\supp{f_{CE}}$ is not the union of any two disjoint affine subspaces of dimension $2$.

Our next claim shows that, up to an invertible linear transformation, 
Construction~\ref{construction:counter} is essentially the only counterexample to the Main Lemma.
\begin{claim}
When $k=4$, either $f$ is the indicator function of a disjoint union of two affine subspaces of dimension $n-k$, or
the Fourier spectrum of $f$ is given by Construction~\ref{construction:counter} under some invertible linear transformation,
and consequently $f$ is the indicator function of a disjoint union of four affine subspaces of dimension $n-k-1$.
\end{claim}
\begin{proof}
When $k=4$, we have $|B|=2^{4-1}-1=7$.
By inequality~\eqref{eqn:sigma_B}, $\sigma[B]=|2B|/|B|\leq 22/7$. But if $|2B|\leq 21$, 
then plugging $K=\sigma[B]\leq 3$ into \eqref{eqn:K} gives that $s \leq 5$ and consequently $F(K)\leq 2K<7$. 
That is, we would have $|\langle B \rangle|< 7|B|=49$. Then following the same argument, we would be able to establish
Lemma~\ref{lemma:2B} for the case $k=4$ as well, 
i.e. to have $|\vspan{B}|=2^k=2(|B|+1)$ and $2B$ is a subspace of dimension $k-1$,
thereby recovering the regular configuration of $f$ being the indicator function of two disjoint affine subspaces of dimension $n-k$.

Therefore, from now on, we assume that $|2B|=22$. 
On the other hand, $|A|=3|B|=21$; combining this with Lemma~\ref{lemma:beta_triangle} 
(i.e. $2B \subseteq A \cup \{\mathbf{0}\}$), we must have $A=2B\setminus\{\mathbf{0}\}$.
By the upper bound on $|\langle B \rangle|$ given in Theorem ~\ref{theorem:Even-Zohar},
we have $|\langle B \rangle| \leq 2^6=64$. 
But if $|\langle B \rangle|<64$ (hence $|\langle B \rangle|=32$ or $|\langle B \rangle|=16$),  
then the proof of Lemma~\ref{lemma:2B} would follow again.

Hence, the counter-example is possible only when the dimension of $\vspan{B}$ is at least $6$. 
Without loss of generality, we may assume $B=\{e_i \mid 1\leq i \leq 6\}\cup\{\beta\}$. 
We will determine vector $\beta$ next.

If $\beta \notin \vspan{e_1,\cdots,e_6}$, then without loss of generality, let $\beta=e_7$.
Now $A=2B\setminus\{\mathbf{0}\}=\{e_i+e_j \mid 1\leq i<j\leq7\}$. 
But applying Proposition~\ref{prop:Boolean} to the vector $e_1+e_2+e_3$ gives that $\hat f(e_1+e_2+e_3)=-6/2^{2k}$,
contradiction to the fact that $\hat f(e_1+e_2+e_3)=0$ because $e_1+e_2+e_3 \notin A \cup B$. 
It follows that $\beta \in \vspan{e_1,\cdots,e_6}$.

Note that every weight-$2$ vector $e_i+e_j$, $1\leq i<j \leq 6$, is in $A$. 
On the other hand, since $|A|= \binom{|B|}{2}$,
it follows that for every $\alpha_k \in A$, there exist a unique pair $\beta_i,\beta_j \in B$ such that $\beta_i+\beta_j=\alpha_k$.
Combining these two facts, 
we conclude that none of the weight-$3$ vector of the form $e_i+e_j+e_k$ is in $B$, for every $1 \leq i < j < k \leq 6$,
as it would gives two ways to obtain vectors such as $e_i+e_j$ by adding two vectors from $B$, thus making $|A|<\binom{|B|}{2}$. 
By Claim \ref{cor:2B-3B}, none of the weight-$4$ vectors can be in $B$ either,  
which leaves only the possibilities of weight-$5$ or weight-$6$ vector for $\beta$.

If $\beta$ is a weight-$5$ vector, without loss of generality, we may assume $\beta= \sum_{i=1}^{5}{e_i}$. 
Then $B$ would contain vectors of weight-$1$ and weight-$5$ only, consequently $A$ would contain vectors of 
weight-$2$, weight-$4$ and weight-$6$ only.
Now applying Proposition~\ref{prop:Boolean} to the vector $e_1+e_2+e_3$ yields $\hat{f}(e_1+e_2+e_3)<0$,
contradicting to the fact that $\hat{f}(e_1+e_2+e_3)=0$ as $e_1+e_2+e_3 \notin A \cup B$. 
Therefore, we have $\beta=\sum_{i=1}^{6}{e_i}$, completing the proof of the claim.
\end{proof}

\section{Proof of the Main Theorem}\label{sec:main_thm}

Clearly, if $\hat{f}(\mathbf{0})=\frac{1}{2^k}$, then, because $|\hat{f}(\alpha)|\leq \hat{f}(\mathbf{0})$ for every $\alpha$,
all non-zero Fourier coefficients of $f$ have absolute value $\frac{1}{2^k}$.
Therefore, Rothschild and van Lint Theorem applies and $f$ is the indicator function of an affine subspace of dimension $n-k$.
Therefore, from now on, we assume $\hat{f}(\mathbf{0})=\frac{1}{2^{k-1}}$.

The first step in our proof of the Main Theorem is to follow a similar procedure
employed in the proof of Theorem~\ref{thm:RV}. That is, whenever possible,
we reduce the values of $n$ and $k$ simultaneously. This proceeds as follows.
Suppose there exists a non-zero $\alpha$ with $\hat{f}(\alpha)=\frac{1}{2^{k-1}}$ or $-\frac{1}{2^{k-1}}$.
Without loss of generality, assume that $\hat{f}(\alpha)=\frac{1}{2^{k-1}}$.
Apply an invertible linear transform $L$ that maps $\alpha$ to $e_1$ and let $g:=Lf$.
Now we have $\hat{g}(\mathbf{0})=\hat{g}(e_1)=\frac{1}{2^{k-1}}$.
Apply the restriction on the first bit of the input to get sub-functions
$g_0$ and $g_1$. Then by \eqref{eqn:subfunction},
$\hat{g}_1(\mathbf{0})=\hat{g}(\mathbf{0})-\hat{g}(e_1)=0$, which implies that $g_1\equiv 0$.
This implies that $\supp{f}$ is completely contained in the support of $g_0$ and moreover, by \eqref{eqn:original_function},
$\hat{g}_0(\beta)=2\hat{f}(0, \beta)$ for every $\beta \in \F_2^{n-1}$.
In other words, $g_0$ is a Boolean function over $\F_2^{n-1}$ and $|\hat{g}(\beta)|$ is equal to either zero, or $\frac{1}{2^{k-1}}$,
or $\frac{1}{2^{k-2}}$. That is, by performing a linear restriction, we reduce both the dimension $n$ and the parameter $k$ by one,
so that the Main Theorem holds for Boolean functions over $\F_2^n$ as long as it holds for Boolean functions over $\F_2^{n-1}$.

When we arrive at a point that such a linear restriction is no longer possible; equivalently, $f$ is irreducible,
then $\hat{f}(\mathbf{0})$ is the
only Fourier coefficient whose absolute value is $\frac{1}{2^{k-1}}$.
Therefore, the Main Lemma for $k\geq 5$ or Lemma~\ref{lemma:small_values} for $2 \leq k\leq 4$ applies.


\section{Concluding Remarks and Open Problems}\label{Sec:conclusion}
In this work, we extend a classical result of Rothschild and van Lint to give a complete characterization of
Boolean functions whose Fourier coefficients take values only in the set $\{-2/2^k, -1/2^k, 0, 1/2^k, 2/2^k\}$.
Our work may be regarded as a first step toward understanding the structures of Boolean functions of granularity $k$.
A major motivation for such studies is to prove a polynomial upper bound on the kill number for any $k$-granular Boolean function,
thus resolving the Log-rank XOR conjecture. Another interesting question is to find 
other sets of Fourier coefficients which uniquely or almost uniquely determine the structures of 
their corresponding Boolean functions.

\section*{Acknowledgments}
We would like to thank anonymous referees for their valuable comments and suggestions which help
us correcting errors, simplifying proofs and improving presentations.
Ning Xie's research was partially supported by grant ARO W911NF1910362.

\bibliographystyle{plain}  
\bibliography{PropertyTest-Bib}

\begin{appendix}

\section{A Proof of Proposition~\ref{prop:restriction}}\label{sec:restriction}
Recall that Proposition~\ref{prop:restriction} on the Fourier spectra of sub-functions obtained from linear restrictions
is the following:
\begin{prop-restriction}
Let $f: \F_2^n \to \R$ be a function defined on the Boolean hypercube.
Let $f_0, f_1: \F_{2}^{n-1} \to \R$ be the ``sub-functions'' obtained from restricting the first bit of the input to $0$ and $1$, respectively;
that is, $f_0(y):= f(0,y)$ and $f_1(y) := f(1, y)$ for all $y\in \F_{2}^{n-1}$.
Then the Fourier spectra of $f_0$ and $f_1$ satisfy that, for all $\beta\in \F_{2}^{n-1}$,
\begin{align}
\hat{f}_{0}(\beta) &= \hat{f}(0,\beta) + \hat{f}(1,\beta), 
\qquad
\hat{f}_{1}(\beta) = \hat{f}(0,\beta) - \hat{f}(1,\beta). \tag{\ref{eqn:subfunction}}
\end{align}
Conversely, the Fourier spectrum of $f$ satisfies
\begin{align}
\hat{f}(0,\beta) &= \frac{1}{2} (\hat{f}_{0}(\beta) + \hat{f}_{1}(\beta)),
\qquad
\hat{f}(1,\beta) = \frac{1}{2} (\hat{f}_{0}(\beta) - \hat{f}_{1}(\beta)). \tag{\ref{eqn:original_function}}
\end{align}
\end{prop-restriction}
\begin{proof}
Clearly it suffices to prove either \eqref{eqn:subfunction} or \eqref{eqn:original_function}
and the other follows immediately.
We prove the first part of \eqref{eqn:original_function}, the second part can be proved analogously. 
By the definition of Fourier transform,
\begin{align*}
\hat{f}(0,\beta)
&= \frac{1}{2^n}\sum_{x\in \F_2^n}f(x)\chi_{(0, \beta)}(x) \\
&= \frac{1}{2^n} \sum_{y\in \F_2^{n-1}}\left(f(0,y)\chi_{(0, \beta)}((0,y)) + f(1,y)\chi_{(0, \beta)}((1,y)) \right) \\
&= \frac{1}{2^n}\left(\sum_{y\in \F_2^{n-1}}f(0,y)\chi_{\beta}(y) + \sum_{y\in \F_2^{n-1}}f(1,y)\chi_{\beta}(y)\right)\\
&= \frac{1}{2^{n}}\sum_{y\in \F_2^{n-1}}f_{0}(y)\chi_{\beta}(y) +\frac{1}{2^{n}}\sum_{y\in \F_2^{n-1}}f_{1}(y)\chi_{\beta}(y) \\
&=  \frac{1}{2} (\hat{f}_{0}(\beta) + \hat{f}_{1}(\beta)). \qedhere
\end{align*}
\end{proof}


\section{The Fourier spectrum of disjoint union of two affine subspaces}\label{sec:Fourier_spectrum}

In this section we calculate the Fourier spectrum of a Boolean function whose support is the union of two
disjoint affine subspaces satisfying certain properties. In particular, the two affine subspaces are of the same dimension
and their Fourier spectra have minimum intersection.

Let $n\geq 1$ and $0\leq k < n$ be integers.
If $V$ is a linear subspace in $\F_2^n$ of dimension $n-k$ and $a\in V^{\perp}$,
where $V^{\perp}$ denotes the linear subspace that is the
\emph{orthogonal complement} of $V$, then it is well known that the Fourier spectrum of the indicator function of
affine subspace $a+V$ is (see e.g.~\cite{Odo14}):
\[
\hat{\mathds{1}}_{a+V}(\alpha)=
\begin{cases}
\frac{1}{2^k}\chi_{\alpha}(a) 	& \text{ if $\alpha \in V^{\perp}$,} \\
 0 								& \text{ otherwise.}
\end{cases}
\]

Let $f:\F_2^n \to \{0,1\}$ be a Boolean function whose support is the union of two disjoint affine subspaces of dimension $n-k$.
By a shift of the origin if necessary, we may assume that one of the two affine subspaces is a linear subspace.
Therefore $f=\mathds{1}_{a+V_1}+\mathds{1}_{V_2}$,
where $V_1$ and $V_2$ are two linear subspaces of dimension $n-k$ in $\F_2^n$ and $a\in V_1^{\perp}$.
In order for $a+V_1$ and $V_2$ to be disjoint, a necessary condition is that their orthogonal complement subspaces
have non-trivial intersection,
$V_1^{\perp} \cap V_2^{\perp} \neq \{\mathbf{0}\}$.
The special configuration we are interested in is when this intersection is minimal, that is
when $|V_1^{\perp} \cap V_2^{\perp}| = 2$.

To this end, without loss of generality, we let $V_1^{\perp}=\vspan{e_1, \ldots, e_k}$ and
$V_2^{\perp}=\vspan{e_k, \ldots, e_{2k-1}}$ so that $V_1^{\perp} \cap V_2^{\perp}=\{\mathbf{0}, e_k\}$.
Then we necessarily have\footnote{This is because, the affine subspace $a+V_1$ can be expressed as the solutions
to a system of linear equations
$a+V_1=\{x\in \F_2^n \mid \langle x, e_i \rangle = a_i \text{ for every $1\leq i \leq k$} \}$,
where $\{e_1, \ldots, e_k\}$ is an orthonormal basis for $V_1^{\perp}$, and
$\{a_i:=\langle e_i, a\rangle\}_{i=1}^{k}$ are the components under this basis.
Now if $|V_1^{\perp} \cap V_2^{\perp}| = 2$, and because the intersection of the two orthogonal complement subspaces is a subspace,
we may take $V_1^{\perp} \cap V_2^{\perp}=\{\mathbf{0}, e_k\}$ for convenience.
On the other hand, $V_2=\{x\in \F_2^n \mid \langle x, e_i \rangle = 0 \text{ for every $k\leq i \leq 2k-1$} \}$.
$a+V_1$ and $V_2$ are disjoint if and only if there is no solution to the two systems of linear equations combined together,
which is equivalent to the condition that $\langle e_k, a\rangle = 1$.}
$\langle e_k, a\rangle = 1$. Therefore for simplicity (and also without loss of generality) we may take $a=e_k$.
Therefore the Fourier spectrum of $f$ is
\begin{align}\label{eqn:Fourier_spectrum}
\hat{f}(\alpha) = \hat{\mathds{1}}_{a+V_1}(\alpha)+\hat{\mathds{1}}_{V_2}(\alpha)=
\begin{cases}
\frac{1}{2^{k-1}} 	& \text{ if $\alpha = \mathbf{0}$,} \\
-\frac{1}{2^k}		& \text{ if $\alpha \in e_k+(\vspan{e_1, \ldots, e_{k-1}}\setminus \{\mathbf{0}\})$,} \\
\frac{1}{2^k}		& \text{ if $\alpha \in \vspan{e_1, \ldots, e_{k-1}}\setminus \{\mathbf{0}\}$,} \\
\frac{1}{2^k}		& \text{ if $\alpha \in e_k+(\vspan{e_{k+1}, \ldots, e_{2k-1}}\setminus \{\mathbf{0}\})$,} \\
\frac{1}{2^k}		& \text{ if $\alpha \in \vspan{e_{k+1}, \ldots, e_{2k-1}}\setminus \{\mathbf{0}\}$,} \\
 0 								& \text{ otherwise.}
\end{cases}
\end{align}

\end{appendix}

\end{document}